\documentclass[conference]{IEEEtran}

\usepackage{mypackage}

\begin{document}

\title{Delay Analysis for Wireless D2D Caching with Inter-cluster Cooperation}

\author{\IEEEauthorblockN{Ramy Amer\IEEEauthorrefmark{1},			
M. Majid Butt\IEEEauthorrefmark{1},
Mehdi Bennis\IEEEauthorrefmark{2}\IEEEauthorrefmark{4}, and
Nicola Marchetti\IEEEauthorrefmark{1}~\IEEEmembership{Fellow,~IEEE}}
\IEEEauthorblockA{\IEEEauthorrefmark{1}CONNECT, Trinity College, University of Dublin, Ireland}
\IEEEauthorblockA{\IEEEauthorrefmark{2}Centre for Wireless Communications, University of Oulu, Finland}
\IEEEauthorblockA{\IEEEauthorrefmark{4}Department of Computer Engineering, Kyung Hee University, South Korea}
\IEEEauthorblockA{email: \{ramyr, majid.butt, nicola.marchetti\}@tcd.ie, bennis@ee.oulu.fi}
\thanks{Manuscript received December 1, 2012; revised August 26, 2015.
Corresponding author: M. Shell (email: http://www.michaelshell.org/contact.html).}}

\maketitle			

\begin{abstract}

Proactive wireless caching and D2D communication have emerged as promising techniques for enhancing users' quality of service and network performance. In this paper, we propose a new architecture for D2D caching with inter-cluster cooperation. We study a cellular network in which users cache popular files and share them with other users either in their proximity via D2D communication or with remote users using cellular transmission. We characterize the network average delay per request from a queuing perspective. Specifically, we formulate the delay minimization problem and show that it is NP-hard. Furthermore, we prove that the delay minimization problem is equivalent to minimization of a non-increasing monotone supermodular function subject to a partition matroid constraint. A computationally efficient greedy algorithm is proposed which is proven to be locally optimal within a factor 2 of the optimum. Simulation results show more than 45\% delay reduction compared to a D2D caching system without inter-cluster cooperation.


\end{abstract}

\IEEEpeerreviewmaketitle

\section{Introduction}
The rapid proliferation of mobile devices has led to an unprecedented growth in wireless traffic demands. A typical approach to deal with such demand is by densifying the network. For example, macrocells and femtocells are deployed to enhance the capacity and attain a good quality of service (QoS), by bringing the network closer to the user. Recently, it has been shown that only a small chunk of multimedia content is highly demanded by most of the users. This small chunk forms the majority of requests that come from different users at different times, which is referred to as $\textit {asynchronous content reuse}$ \cite{mono}. 

Caching the most popular content at network edges has been proposed to avoid serving all requests from the core network through highly congested backhaul links \cite{wr}. From the caching perspective, there are three main types of networks, namely, caching on femtocells in small cell networks, caching on remote radio heads (RRHs) in cloud radio access networks (RANs), and caching on mobile devices \cite{femtocell_mehdi, ran, D2D}. In this paper, we focus on device caching solely. The architecture of device caching exploits the large storage available in modern smartphones to cache multimedia files that might be highly demanded by users. User devices exchange multimedia content stored on their local storage with nearby devices \cite{D2D}. Since the distance between the requesting user and the caching user (a user who stores the file) will be small in most cases, device to device (D2D) communication is commonly used for content transmission \cite{D2D}.

In \cite{D2D1}, a novel architecture is proposed to improve the throughput of video transmission in cellular networks based on the caching of popular video files in base station controlled D2D communication. The analysis of this network is based on the subdivision of a macrocell into small virtual clusters, such that one D2D link can be active within each cluster. Random caching is considered where each user caches files at random and independently, according to a caching distribution.

In \cite{D2D2}, the authors studied joint optimization of cache content placement and scheduling policies to maximize the offloading gain of a cache-enabled D2D network. In \cite{D2D3}, the authors proposed an opportunistic cooperation strategy for D2D transmission by exploiting the caching capability at the users to control the interference among D2D links.

Our work is motivated by a D2D caching framework studied in \cite{D2D,D2D1,D2D2, D2D3} to propose and analyze a new D2D caching architecture with inter-cluster cooperation. We show that allowing inter-cluster D2D communication via cellular link helps to reduce the network average delay per request. To the best of our knowledge, none of the works in the literature dealt with the delay analysis of D2D caching networks with inter-cluster cooperation. The main contributions of this paper are summarized as follows:
\begin{itemize}
\item We study a D2D caching system with inter-cluster cooperation from a queueing perspective. We formulate the network average delay minimization problem in terms of cache placement and show that it is NP-hard.
\item A closed form expression of the network average delay is derived under the policy of caching popular files. Moreover, a locally optimal greedy caching algorithm is proposed within a factor of 2 of the optimum. Results show that the delay can be significantly reduced by allowing D2D caching with inter-cluster cooperation.
\end{itemize} 

The rest of the paper is organized as follows. The system model is presented in Section II. In Section III, we formulate the problem and perform delay analysis of the system. In Section IV, two content caching schemes are studied. Finally, we discuss the simulation and analytic results  in Section V and conclude the paper in Section VI.

\section{System Model}
\label{sysmodel}
\begin{figure}
\centering
\includegraphics[width=0.4\textwidth]{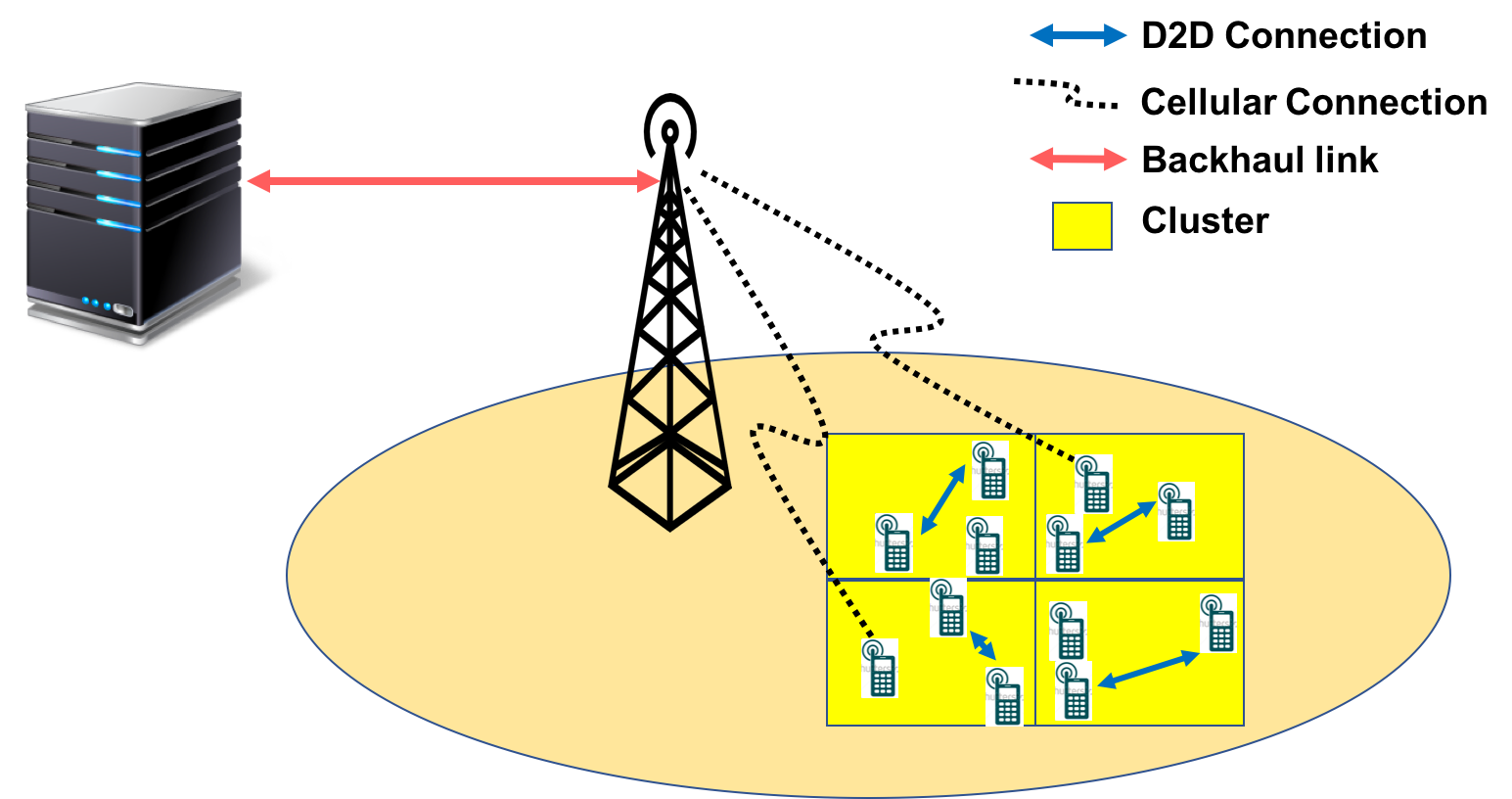}				
\caption {Schematic diagram of the proposed system model.}
\label{Network Model}
\end{figure}

\subsection{Definitions}
In Table \ref{notation}, we summarize some of the definitions extensively used in this paper.
\begin{table}[!h]
\caption{List of notations.}
    \label{notation}
\begin{center}
 \begin{tabular}{ |c|l| }
  \hline
  Notation & Description \\ \hline
$\mathcal{U}$  &Set of users of size $U$\\ \hline			
$\mathcal{F}$ &Library of files of size $F$\\ \hline			 
 $\mathcal{K}$ &Set of clusters of size $K$\\ \hline		 
   $R_{D}$     &Data rate of D2D transmission \\ \hline
    $R_{WL}$ &Data rate of inter-cluster communication using \\ & cellular network \\ \hline
     $R_{BH}$ &Data rate of backhaul transmission \\ \hline
 \end{tabular}
\end{center}
\end{table}	
\subsection{Network Model}
In this subsection, we describe our proposed D2D caching network with inter-cluster cooperation. Fig.~\ref{Network Model} illustrates the system layout. A cellular network is formed by a base station (BS) and a set of users $\mathcal{U}= \{1,\dots, U\}$ placed uniformly in the cell. The cell is divided into $K$ small clusters of equal size, with $U/K$ users per cluster. Inside each cluster, users can communicate directly using high data rate D2D communication in a dedicated frequency band for D2D transmission.

Each user $u \in \mathcal{U}$ makes a request for file $ f \in \mathcal{F}$ in an independent and identical manner according to a given request probability mass function. It is assumed that each user can cache up to $M$ files, and for the caching problem to be non-trivial, it is assumed that $M \leq F$. From the cluster perspective, we assume to have a cluster's virtual cache center ($VCC$) formed by the union of devices' storage in the same cluster which caches up to $N$ files, i.e., $N = (U/K)M$.

 We define three modes of operation according to how a request for content $ f \in \mathcal{F}$ is served:
\begin{enumerate}
\item {\textbf{Local cluster mode ($M_{lc}$ mode):}} Requests are served from the local cluster.  Files are downloaded from nearby users via a one-hop D2D connection. In this mode, we neglect self-caching, i.e., when a user finds the requested file in its internal cache with zero delay.
\item \textbf{Remote cluster mode ($M_{rc}$ mode):} Requests are served from any of the remote clusters via inter-cluster cooperation. The BS fetches the requested content from a remote cluster, then delivers it to the requesting user by acting as a relay in a two-hop cellular transmission.
\item \textbf{Backhaul mode ($M_{bh}$ mode):} Requests are served directly from the backhaul. The BS obtains the requested file from the core network over the backhaul link and then transmits it to the requesting user.
\end{enumerate}

To serve a request for file $f$ in cluster $k \in \mathcal{K}$, first, the BS searches the $VCC$ of cluster $k$. If the file is cached, it will be delivered from the local $VCC$ ($M_{lc}$ mode). We assume that the BS has all information about cached content in all clusters, such that all file requests are sent to the BS, then the BS replies with the address of caching user from whom the file will be received. We use an interference avoidance scheme, where at most one transmission is allowed in each cluster on any time-frequency slot. If the file is not cached locally in cluster $k$ but cached in any of the remote clusters, it will be fetched from a randomly chosen cooperative cluster ($M_{rc}$ mode), instead of downloading it from the backhaul. Unlike multi-hop D2D cooperative caching discussed in \cite{multihop}, in our work, cooperating clusters use 2-hop high rate cellular connections to exchange cached files, such that the D2D band is dedicated only to the intra-cluster communication. Hence, all the inter-cluster communication is performed in a centralized manner through the BS. Finally, if the file has not been cached in any cluster $j \in \mathcal{K}$ in the cell, it can be downloaded from the core network through the backhaul link ($M_{bh}$ mode).


\subsection{Content Placement and Traffic Characteristics}
\label{sec}
We use a binary matrix $\textbf{C}= [c_{k,f}]_{K\times F}$ with $c_{k,f} \in \{0, 1\}$ to denote the cache placement in all clusters, where $c_k,_f = 1$ indicates that content $f$ is cached in cluster  $k$. 
Fig. \ref{vcc}  shows the $VCC$ of a cluster $k$ modeled as a \textit{multiclass processor sharing queue} (MPSQ) with arrival rate $\lambda_k$  and three serving processors, representing the three transmission modes. According to the MPSQ definition \cite{MPSQ}, each transmission mode is represented by an M/M/1 queue with Poisson arrival rate and exponential service rate.

If a user in cluster $k$ requests a locally cached file $f$ (i.e., $c_k,_f = 1$), it will be served by the local cluster mode with an average rate $R_{D}$. However, if the requested file is cached only in any of the remote clusters (i.e., $c_k,_f = 0$  and $\sum_{j\in \mathcal{K}\setminus \{k\}}c_j,_f \geq 1$), it will be served by the remote cluster mode. We assume the uplink and the downlink transmissions have the same rate $2R_{WL}$, and the BS relays a file between cooperating clusters with an average sum rate $R_{WL}$, which is shared between clusters simultaneously served by the remote cluster mode. Finally, requests for files that are not cached in the entire cell (i.e., $\sum_{j=1}^{K}  c_k,_f = 0$) are served by the backhaul mode with an average sum rate $R_{BH}$. Similarly, $R_{BH}$ is shared between clusters simultaneously served by the backhaul mode. $R_{BH}$ accounts for the whole path rate from the core network to the user through the BS. Additionally, it is assumed that the wireless connection from the users to the BS in the backhaul mode is performed on a frequency band different from the inter-cluster communication band.

\begin{figure}
\begin{center}
\includegraphics[width=0.3\textwidth]{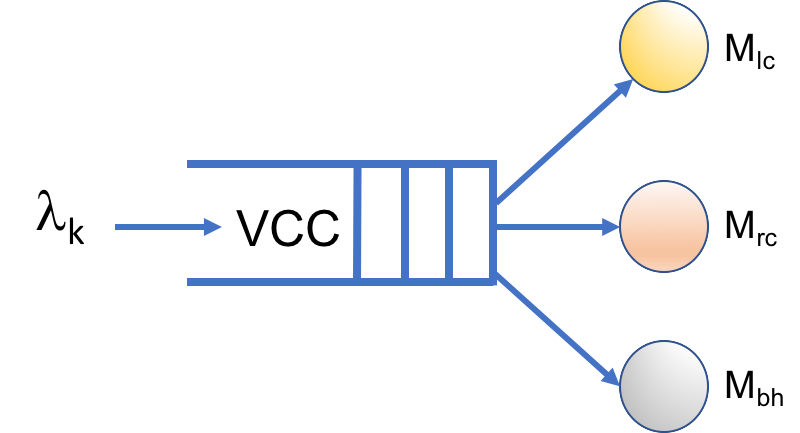}			
\caption {Virtual cache center of cluster $k$ with request arrival rate $\lambda_k$ and three processors (modes) to serve the incoming requests. }
\label{vcc}
\end{center}
\end{figure}

Due to traffic congestion in the core network and the transmission delay between cooperating clusters, we assume that the aggregate transmission rates for the above three modes are ordered such that $R_D > R_{WL} >  R_{BH}$. Moreover, we assume that the content size $S_f$ is exponentially distributed with mean $\overline{S}$  bits. Hence, the corresponding request service time of the three modes also follow an exponential distribution with means $\tau_{lc} = \frac{\overline{S}}{R_{D}}$ sec/request, $\tau_{rc} = \frac{\overline{S}}{R_{WL}}$ sec/request, and $\tau_{bh} = \frac{\overline{S}}{R_{BH}}$ sec/request, respectively.
\section{Problem Formulation}
\label{prob}
In this section, we characterize the network average delay on a per-request basis from the global network perspective.
Specifically, we study the request arrival rate and the traffic dynamics from a queuing perspective and get a closed form expression
for the network average delay.
 \subsection{File Popularity Distribution}
We assume that the popularity distribution of files in all clusters follows Zipf's distribution with skewness order $\beta$ \cite{zipf}. However, it is assumed that the popular files are different from one cluster to another. Our assumption for the popularity distribution is extended from \cite{basic_principle}, where the authors explained that the scaling of popular files is sublinear with the number of users.\footnote
{The number of popular files increases with the number of users with a rate slower than the linear polynomial rate, e.g. the logarithmic rate.}

  To illustrate, if the first user is interested in a set of files of size $m_0$, the second user is also interested in $m_0$ files, then $m_0/2$ files are common with the first user and $m_0/2$ are new files. The third user shares $2m_0/3$ files with the first two users and has $m_0/3$ new files, etc. The union of all highly demanded files by $n$ users is $m = m_0 (1 + \frac{1}{2} + \frac{1}{3} + \dots)= m_0 \sum_{i=1}^{n}\frac{1}{i} \approx m_0$ log $n$. Hence, the library size increases sublinearly with the number of users. In this work, we assume the scaling of the library size is sublinear with the number of clusters. The cell is divided into small clusters with a small number of users per cluster, such that users in the same cluster are assumed to have the same file popularity distribution.

The probability that a file $f$ is requested in cluster $k$, with $m_0$ highly demanded files in each cluster, follows a Zipf distribution written as \cite{zipf},
\begin{equation}
P_k,_f = \frac{(  f - \frac{k-1}{k}m_0 a  +  (F - \frac{k-1}{k}m_0)b  )^{-\beta}}{\sum_{i=1}^{F}i^{-\beta}},
\label{popularity eqn}
\end{equation}
where $a = \textbf{1}(f>\frac{k-1}{k}m_0)$ and $b = \textbf{1}(f \leq\frac{k-1}{k}m_0)$, $\frac{k-1}{k}m_0$ is the order of the most popular file in the $k$-th cluster, and $\textbf{1}(.)$ is the indicator function.
As an example, for the first cluster, $P_1,_f = \frac{(  f  )^{-\beta}}{\sum_{i=1}^{F}i^{-\beta}}$, which is the Zipf's distribution with the most popular file $f = 1$. For example, if $m_0 = 60$, $P_2,_f = \frac{(  f  - 30a + (F - 30)b )^{-\beta}}{\sum_{i=1}^{F}i^{-\beta}}$ for the second cluster, which is the Zipf's distribution with the most popular file $f = \frac{m_0}{2} + 1 = 31$, correspondingly, $f = \frac{2m_0}{3} + 1 = 41$ is the most popular file in the third cluster, and so on.

\subsection{Arrival and Service Rates}
The arrival rates for communication modes $M_{lc}$, $M_{rc}$ and $M_{bh}$ in cluster $k$ are denoted by $\lambda_{k,lc}$, $\lambda_{k,rc}$ and $\lambda_{k,bh}$, respectively while the corresponding service rates are represented by $\mu_{lc}$, $\mu_{rc}$ and $\mu_{bh}$. For the local cluster mode, we have
\begin{equation}
\lambda_k,_{lc} = \lambda_k \sum_{f=1}^{F}  P_k,_f c_k,_f,
\end{equation}
where $\sum_{f=1}^{F}  P_k,_f c_k,_f$ is the probability that the requested file is cached locally in cluster $k$. The corresponding service rate is $\mu_{lc} = \frac{1}{\tau_{lc}}$. For the remote cluster mode, the request arrival rate is defined as,
\begin{equation}
\lambda_k,_{rc} = \lambda_k \sum_{f=1}^{F}  P_k,_f (1 - c_k,_f)\mathrm{min}\Big(\sum_{m \in \mathcal{K}\setminus \{k\}}c_m,_f,1\Big), 
\end{equation}
where min$(\sum_{m \in \mathcal{K}\setminus \{k\}}c_m,_f,1)$ equals one only if the content $f$ is cached in at least one of the remote clusters. Hence, $\sum_{f=1}^{F}  P_k,_f (1 - c_k,_f)$min$(\sum_{m \in \mathcal{K}\setminus \{k\}}c_m,_f,1)$ is the probability that the requested file $f$ is cached in any of the remote clusters given that it was not cached in the local cluster $k$. The corresponding service rate is $\mu_{rc} = \frac{1}{\tau_{rc} N_a}$. $N_a$ represents the number of cooperating clusters simultaneously served by the remote cluster mode, i.e, the number of users which share the cellular rate.  		

Finally, for the backhaul mode, the request arrival rate is written as,
\begin{equation}
\lambda_k,_{bh} = \lambda_k \sum_{f=1}^{F}  P_k,_f  \prod \limits_{k=1}^{K} (1 - c_k,_f),
\end{equation}
where $\sum_{f=1}^{F}  P_k,_f  \prod \limits_{k=1}^{K} (1 - c_k,_f)$ is the probability that the requested file $f$ is not cached in the entire cell, hence this content could be downloaded only from the core network. The corresponding service rate is $\mu_{bh} = \frac{1}{\tau_{bh} N_b}$. Similarly, $N_b$ is defined as the number of clusters simultaneously served by the backhaul mode.

The traffic intensity of a queue is defined as the ratio of mean service time to mean inter-arrival time. We introduce $\rho_k$ as a metric of the traffic intensity at cluster $k$,			

\begin{equation}
\rho_k = \frac{\lambda_k,_{lc}}{\mu_{lc}} +  \frac{\lambda_k,_{rc}}{\mu_{rc}}  +  \frac{\lambda_k,_{bh}}{\mu_{bh}}	
\end{equation}
Similar to \cite{delayequation}, we consider $\rho_k < 1 $ as the stability condition; otherwise, the overall delay will be infinite. The traffic intensity at any cluster is simultaneously related to the request arrival rate and the transmission rates of the three modes.

\subsection{Network Average Delay}	
The average delay per request for cluster $k$ can be defined as \cite{delayequation},
\begin{eqnarray}
D_k &= \frac{\rho_k}{\lambda_k} + \frac{\frac{\lambda_k,_{lc}}{\mu_{lc}^2} +  \frac{\lambda_k,_{rc}}{\mu_{rc}^2}  +  \frac{\lambda_k,_{bh}}{\mu_{bh}^2}}{1 - \rho_k}
         \label{T eqn}
 \end{eqnarray}
Based on the analysis of the delay in a single cluster, we can derive the network weighted average delay per request as,
\begin{eqnarray}
D =  \frac{1}{\lambda} \sum_{k=1}^{K} \lambda_k D_k,
\label{delay equation}
 \end{eqnarray}
where $\lambda =  \sum_{i=1}^{K}\lambda_i$ denotes the overall user request arrival rate in the cell. We can observe that the network average delay depends on the arrival rates of the three transmission modes, which are, in turn, functions of the content caching scheme.
Because of the limited caching capacity on mobile devices, we would like to optimize the cache placement in each cluster to minimize the network weighted average delay per request. The delay optimization problem is formulated as,

\begin{align}
\label{optimize eqn}
& \underset{c_k,_f}{\text{minimize}} \quad D \\
&\textrm{subject to}\quad  \sum_{f=1}^{F} c_k,_f \leq N,\quad c_k,_f \in \{ 0, 1\},
\label{optimize eqn1}
\end{align}		
where (\ref{optimize eqn1}) is the constraint that the maximum cache size is $N$ files per cluster, and the file is either cached entirely or is not cached, i.e., no partial caching is allowed. The objective function in (\ref{optimize eqn}) is not a convex function of the cache placement elements $c_k,_f \in \{ 0, 1\}$. Moreover, this equation can be reduced to a well- known $0-1$ knapsack problem which is already proven to be NP-hard in \cite{NP-hard}.

In the next section, we analyze the network average delay under several caching policies. Moreover, we reformulate the optimization problem in (\ref{optimize eqn}) as a well-known structure that that has a locally optimal solution within a factor of 2 of the optimum.

\section{Proposed Caching Schemes} 	
\label{Caching Schemes}
 \subsection{Caching Popular Files (CPF)}
The most popular files in each cluster are independently stored in their local cluster caches \cite{multi-cell}. Here, we assume that the request arrival rate $\lambda_k$ is equal for all clusters.

\subsubsection{Arrival Rate for D2D Communication}
The arrival rate of the D2D communication mode is written as,
 \begin{equation}
\lambda_k,_{lc} = \lambda_k \sum_{f=\frac{k-1}{k}m_0 + 1}^{\frac{k-1}{k}m_0 + N}  P_k,_f, 	
\label{lamda1}
\end{equation}
where $\sum_{f=\frac{k-1}{k}m_0 + 1}^{\frac{k-1}{k}m_0 + N}  P_k,_f$ is the probability that the requested file is cached in the local cluster $k$, and $ f=\frac{k-1}{k}m_0 + 1$ is the most popular file index for cluster $k$. As an example, for the first cluster, $\lambda_1,_{lc} = \lambda_1 \sum_{f=1}^{N}  P_1,_f $. 
\subsubsection{Arrival Rate for Inter-cluster Communication}
\label{NaNb}	
The arrival rate of the inter-cluster communication mode is given by,
  \begin{equation}
  \lambda_k,_{rc} = \lambda_k \sum_{j \in \mathcal{K}\setminus \{k\}}\sum_{f=c}^{\frac{j-1}{j}m_0 + N}  P_k,_f,    
  \label{lamda2}
\end{equation}
where $c $ is defined as max$\big(\frac{j-2}{j-1}m_0 + N + 1, \frac{j-1}{j}m_0 + 1\big)$, such that the cached content in remote clusters is counted only once when calculating $\lambda_k,_{rc}$ (even if the content is cached in many clusters). The probability that the requested content is cached in any of the remote clusters, while non-existing in the local cluster, can be written as $\sum_{j \in \mathcal{K}\setminus \{k\}}E(c_{k,j,f})$, where
\begin{equation}
E(c_{k,j,f}) = \sum_{f=c}^{\frac{j-1}{j}m_0 + N}  P_k,_f
\end{equation}
$E(c_{k,j,f})$ is the probability that a file $f$ requested by a user in cluster $k$ is cached in a remote cluster $j\neq k$, and is not replicated in other clusters.														

To compute the service rate $\mu_{rc}$, first we need to calculate the number of cooperating clusters $N_a$, since they share the cellular rate. As introduced in Section \ref{prob}, $N_a$ is a random variable representing the number of clusters served by the cellular links whose mean is given by,	
\begin{equation}
 \overline{N_a} =K \frac{\lambda_k,_{rc}}{ \lambda_k} =  K\sum_{j \in \mathcal{K}\setminus \{k\}}\sum_{f=c}^{\frac{j-1}{j}m_0 + N}  P_k,_f
 \end{equation}

 \subsubsection{Arrival Rate for Backhaul Communication}
The arrival rate of the backhaul communication mode is defined as,
  \begin{align}
  \label{lamda3}
  \lambda_k,_{bh}& = \lambda_k\big(1- (\lambda_k,_{lc} + \lambda_k,_{rc})\big)\nonumber \\
  &= \lambda_k\Big(1- \big( \sum_{f=\frac{k-1}{k}m_0 + 1}^{\frac{k-1}{k}m_0 + N}  P_k,_f  +  \sum_{j \in \mathcal{K}\setminus \{k\}}\sum_{f=c}^{\frac{j-1}{j}m_0 + N}  P_k,_f \big)\Big)
\end{align}
Similar to $N_a$, $N_b$ is a random variable representing the number of clusters served by the backhaul link whose mean is given by,
\begin{align}
   \overline{N_b} &= K \frac{\lambda_k,_{bh}}{ \lambda_k} = K\Big(1- \big( \nonumber \\
  &\sum_{f=\frac{k-1}{k}m_0 + 1}^{\frac{k-1}{k}m_0 + N}  P_k,_f  +  \sum_{j \in \mathcal{K}\setminus \{k\}}\sum_{f=c}^{\frac{j-1}{j}m_0 + N}  P_k,_f \big)\Big)  
\end{align}
Obviously, we have $ \lambda_k = \lambda_k,_{lc}+\lambda_k,_{rc}+\lambda_k,_{bh}$. From (\ref{lamda1}), (\ref{lamda2}), and (\ref{lamda3}), the network average delay can be calculated directly from (\ref{delay equation}). CPF in each cluster is computationally straightforward if the highly demanded content is known. Additionally, it is easy to implement in an independent and distributed manner. However, CPF achieves high performance only if the popularity exponent $\beta$ is large enough, since a small chunk of content is highly demanded, which can be cached entirely in each cluster.

 \subsection{Greedy Caching Algorithm (GCA)}	
 \label{mat}
In this subsection, we introduce a computationally efficient caching algorithm. We prove that the minimization problem in (\ref{optimize eqn}) can be written as a \textit{minimization of a supermodular function} subject to \textit{partition matroid constraints}. This structure has a greedy solution which has been proven to be local optimal within a factor 2 of the optimum \cite{multi-cell}, \cite{solnmono2}.

We start with the definition of supermodular and matroid functions, then we introduce and prove some relevant lemmas.
\subsubsection{Supermodular Functions}
Let $\mathcal{S}$ be a finite ground set. The power set of the set $\mathcal{S}$ is the set of all subsets of $\mathcal{S}$, including the empty set and $\mathcal{S}$ itself. A set function $f$, defined on the powerset of $\mathcal{S}$ as $f$: $2^\mathcal{S}$$\to R$, is supermodular if for any $A \subseteq B \subseteq \mathcal{S}$ and $x \in \mathcal{S}\setminus B$ we have \cite{solnmono2}
\begin{equation}
f(A\cup \{x\}) - f(A) \leq f(B\cup \{x\}) - f(B)
\end{equation}
To explain, let $f_A(x) = f(A\cup x) - f(A)$ denote the marginal value of an element $x \in \mathcal{S}$ with respect to a subset $A \subseteq \mathcal{S}$. Then, $\mathcal{S}$ is supermodular if for all $A \subseteq B \subseteq \mathcal{S}$ and for all $x \in \mathcal{S}\setminus B$ we have $f_A(x) \leq f_B(x)$, i.e., the marginal value of the included set is lower than the marginal value of the including set \cite{ solnmono2}.

 \subsubsection {Matroid Functions}
  Matroids are combinatorial structures that generalize the concept of linear independence in matrices \cite{solnmono2}. A matroid $\mathcal{M}$ is defined on a finite ground set $\mathcal{S}$ and a collection of subsets of $\mathcal{S}$ said to be independent. The family of these independent sets is denoted by $\mathcal{I}$ or $\mathcal{I}(\mathcal{M})$. It is common referring to a matroid $\mathcal{M}$ by listing its ground set and its family of independent sets, i.e., $\mathcal{M} = (\mathcal{S}, \mathcal{I})$. For $\mathcal{M}$ to be a matroid, $\mathcal{I}$ must satisfy these three conditions:
\begin {itemize}
\item $\mathcal{I}$ is nonempty set.
\item $\mathcal{I}$ is downward closed; i.e., if $B \in \mathcal{I}$ and $A \subseteq B$, then $A \in \mathcal{I}$.
\item If $A$ and  $B$ are two independent sets of  $\mathcal{I}$ and $B$ has more elements than $A$, then $\exists{e} \in B\setminus A$ such that $A \cup\{e\}\in \mathcal{I}$.

\end{itemize}
One special case is a partition matroid in which the ground set $\mathcal{S}$ is partitioned into disjoint sets $\{S_1, S_2,\dots , S_l\}$, where
 \begin{equation}
 \label{mat defn eqn}
\mathcal{I} = \{A\subseteq \mathcal{S}: |A\cap S_i|\leq k_i  \textrm{ for  all}\ i = 1, 2,\dots, l\},
\end{equation}
for some given integers $k_1, k_2, \dots, k_l$.

\begin{lemma}
The constraints in equation (\ref{optimize eqn1}) can be written as a partition matroid on a ground set that characterizes the caching elements on all clusters.
\end{lemma}
\begin{proof}
See Appendix A.
\end{proof}

\begin{lemma}
The objective function in equation (\ref{optimize eqn}) is a monotone non-increasing supermodular function.
\end{lemma}
\begin{proof}
See Appendix B.
\end{proof}
The greedy solution for this problem structure has been proven to be locally optimal within a factor 2 of the optimum \cite{multi-cell}, \cite{solnmono2}.
The greedy caching algorithm for the proposed D2D caching system with inter-cluster cooperation is  illustrated in Algorithm 1.

\begin{algorithm}
    \SetKwInOut{Input}{Input}
    \SetKwInOut{Output}{Output}

    \Input{$K$, $F$, $N$, $\beta$, $\overline{S}$, $R_D$, $\overline{R_{WL}}$, $\overline{R_{BH}}$;}
    \textbf{Initialization}: {$C \gets (0)_{K\times F}$}\;
    \tcc{Check if all clusters are fully cached.}
       \While{$\sum_{k=1}^{K} \sum_{f=1}^{F}c_{k,f} < NK$}		
      {
        ($k^*$, $f^*$) $\gets \mathrm{argmax}_{(k, f)} D(C) - D(C\cup S_k^f)$\;
        \tcc{File achieving highest marginal value is cached.}
        $c_{k^*,f^*}=1$ \;

      }
      \Output{Cache placement $C$;}
 \caption{Greedy Caching Algorithm}
\end{algorithm}



\section{Numerical Results}
\label{sim}
In this section, we evaluate the performance of our proposed inter-cluster cooperative architecture using simulation and analytic results. Results are obtained with the following parameters: $\lambda_k = 0.5$ requests/sec, $m_0=60$ files, $F = 108$ files, $\overline{S} = 4$ Mbits, $K = 5$ clusters, $U = 25$ users, $M=4$ files, and $N= 20$ files. $R_{WL} =50$ Mbps and $R_{BH} = 5$ Mbps  as in \cite{multi-cell}. For a typical D2D communication system with transmission power of 20 dBm, transmission range of 10 m, and free space path loss as in \cite{basic_principle}, we have $R_{D} = 120$ Mbps.

\begin{figure}
\begin{center}
\includegraphics[width=3.2 in]{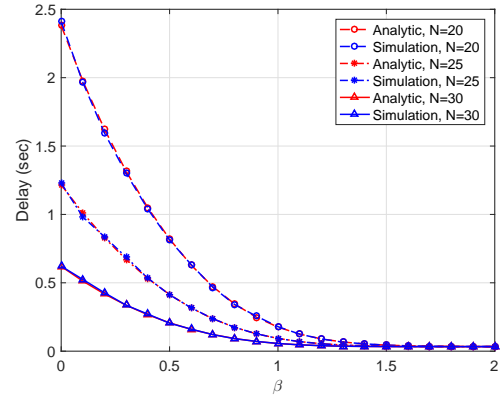}  
\caption {Network average delay versus popularity exponent under CPF scheme. }
\label{delay}
\vspace{-0.4cm}
\end{center}
\end{figure}


\begin{figure*}
\centering
  \subfigure[Average delay (plotted against the left hand side y-axis) and gain (plotted against the right hand side y-axis) vs cluster cache size under CPF scheme.]{\includegraphics[width=2.3in]{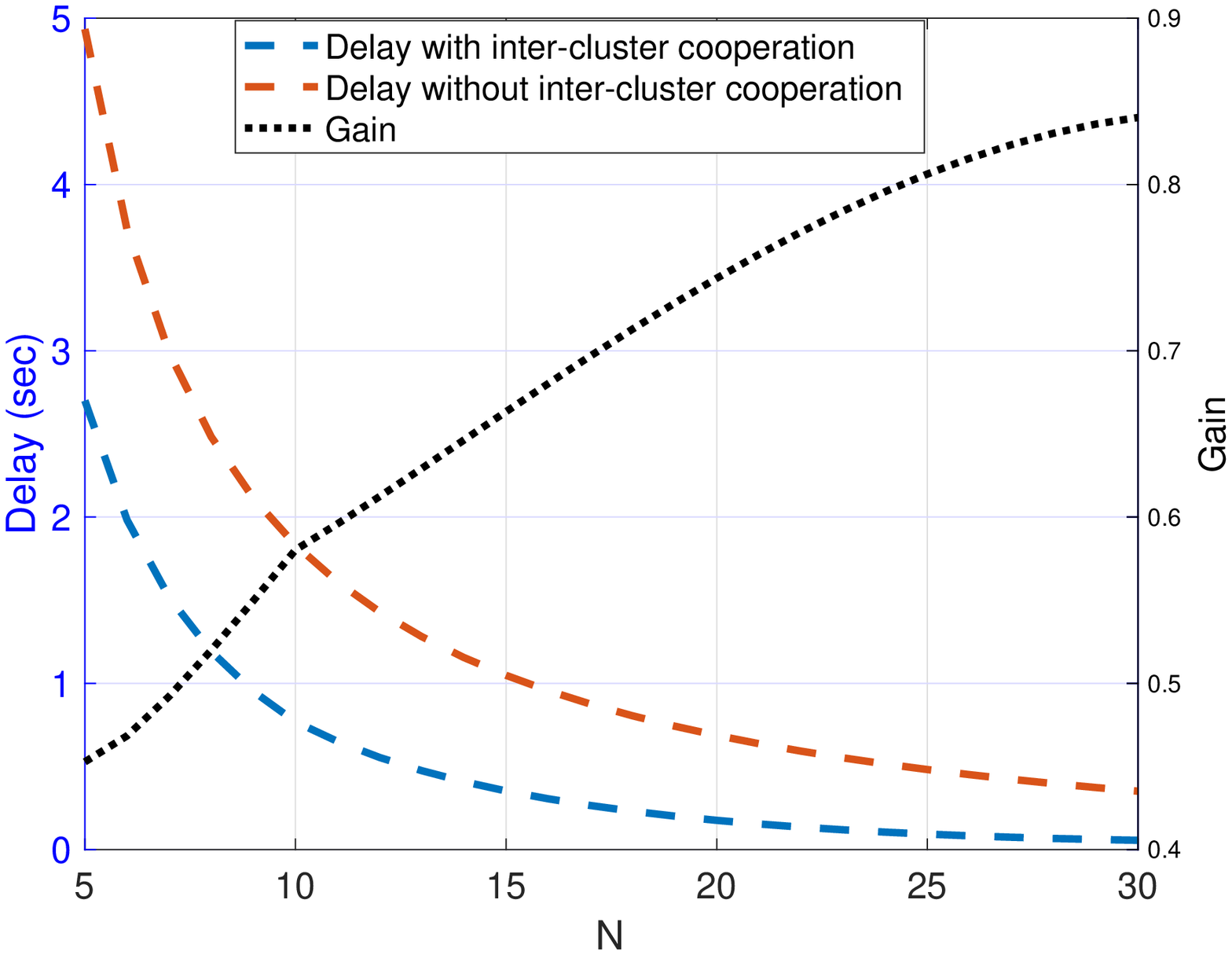}
\label{cache_size}}
\subfigure[Average delay vs request arrival rate under three caching schemes, CPF, GCA, and RC ($R_{D}= 50$ Mbps, $\overline{R_{WL}}= 15$ Mbps, $\overline{R_{BH}}= 10$Mbps, $\beta=0.5$).]{\includegraphics[width=2.3in]{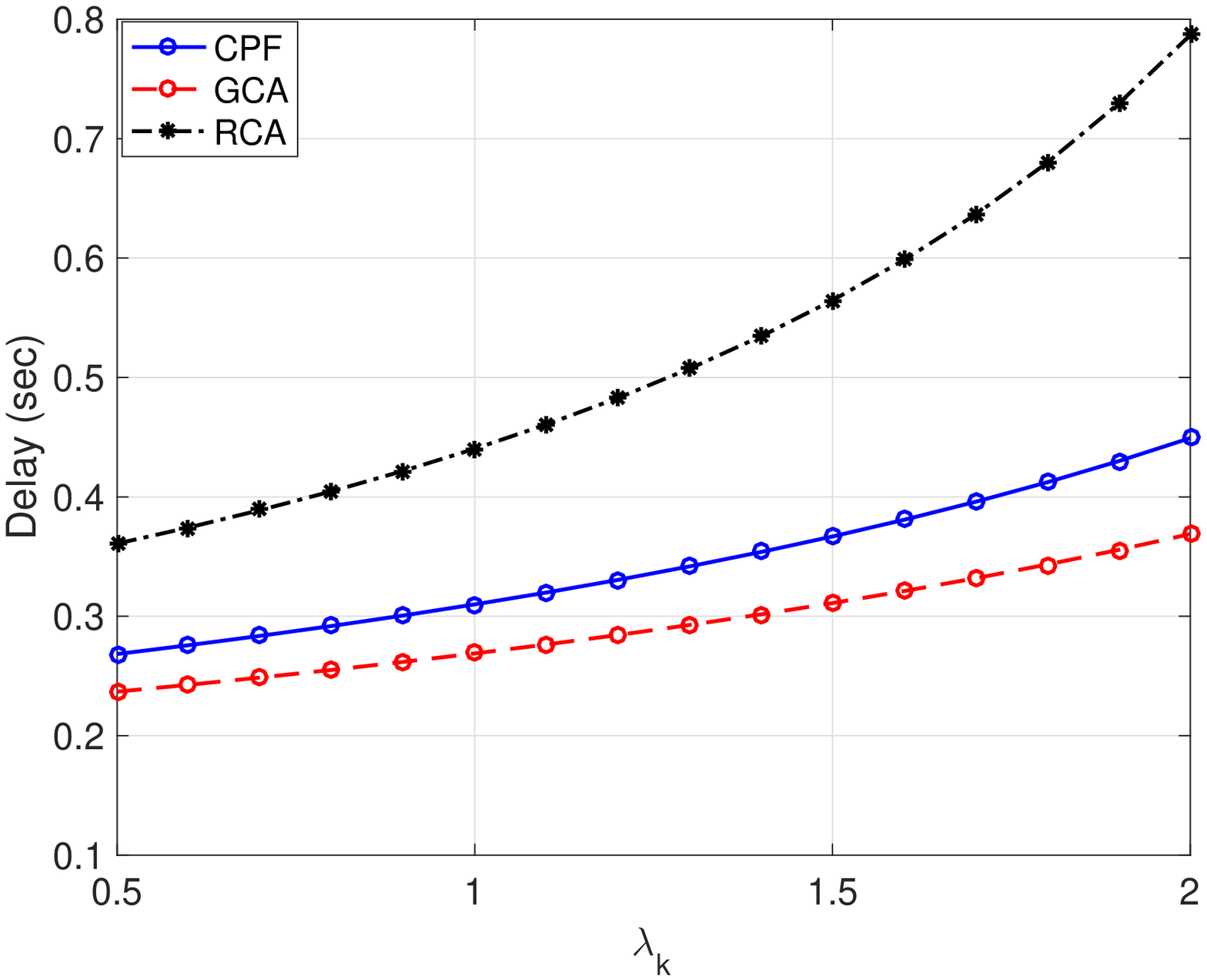}		
  \label{locally optimal delay vs lambda}}
\subfigure[Average delay vs popularity exponent for three caching schemes, CPF, GCA, and RC ($R_{D}= 50$ Mbps, $\overline{R_{WL}}= 15$ Mbps, $\overline{R_{BH}}= 10$Mbps, $\lambda_k = 0.5$ requests/sec).]{\includegraphics[width=2.3in]{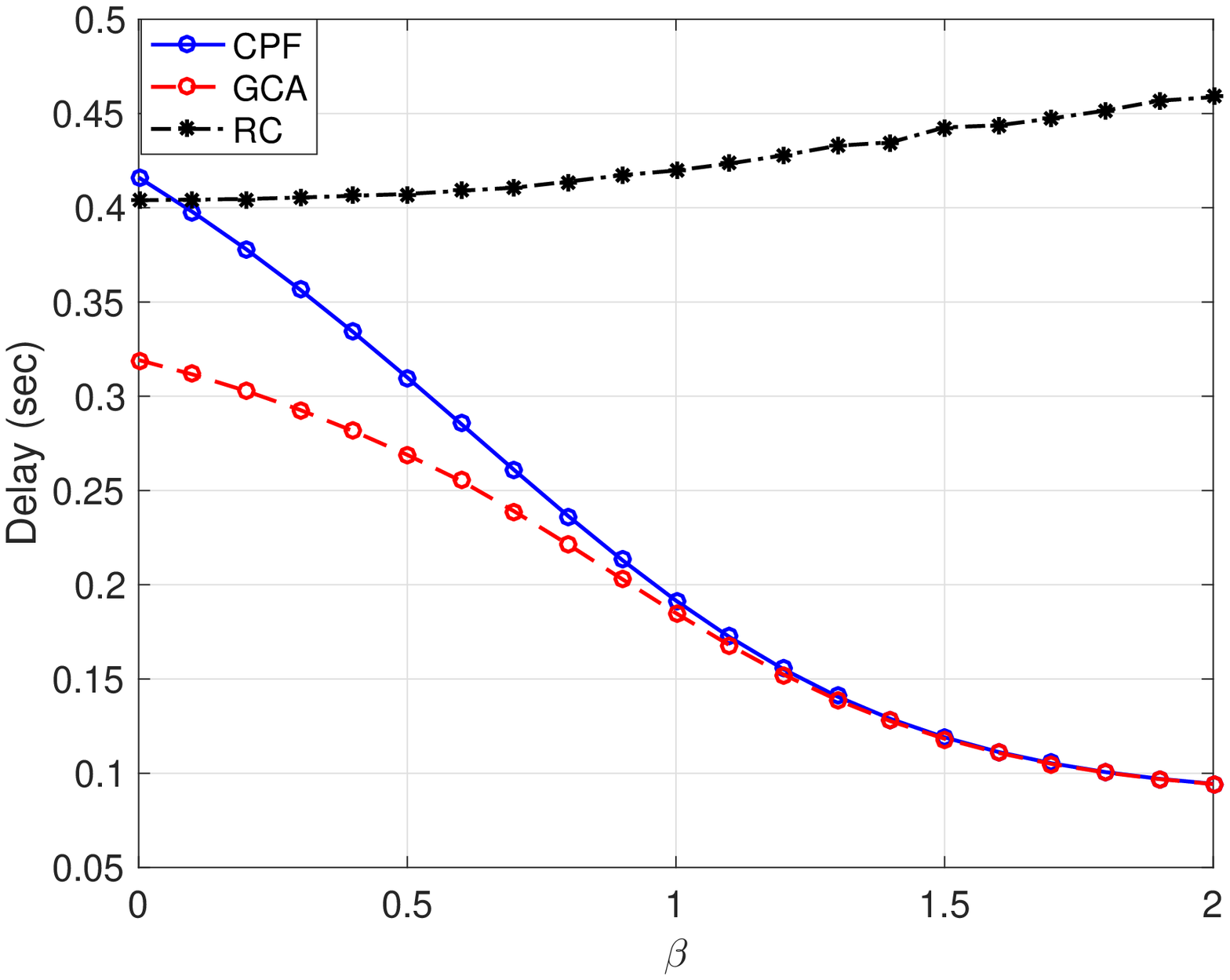}		
  \label{CPF and GCA}}
\caption{Evaluation and comparison of average delay for the proposed schemes and random caching for various system parameters.}	
\label{all}
\end{figure*}

In Fig.~\ref{delay}, we verify the accuracy of the analytic results of the network average delay under CPF with inter-cluster cooperation. The theoretical and simulated results for the network average delay under CPF scheme are plotted together, and they are consistent. We can see that the network average delay is significantly improved by increasing the cluster cache size $N$. Moreover, as $\beta$ increases, the average delay decreases. This is attributed to the fact that a small chunk of content forms most requests, which can be cached locally in each cluster and delivered via high data rate D2D communication.

In Fig.~\ref{all}, we evaluate and compare the performance of various caching schemes.  In Fig.~\ref{cache_size}, our proposed inter-cluster cooperative caching system is compared with a D2D caching system without cooperation under CPF scheme. For a D2D caching system without cooperation, requests for files that are not cached in the local cluster are downloaded directly from the core network. Fig.~\ref{cache_size} shows that, for a small cluster cache size, the delay reduction (gain) of our proposed inter-cluster cooperative caching is higher than 45\% with respect to a D2D caching system without inter-cluster cooperation and greater than 80\% if the cluster cache size is large. We define the delay reduction gain as
 \begin{equation}
\textrm{Gain} = 1 - \frac{\textrm{Delay with inter-cluster cooperation}}{\textrm{Delay without inter-cluster cooperation}}	
\end{equation}

Fig.~\ref{locally optimal delay vs lambda} shows the network average delay plotted against request arrival rate $\lambda_k$ for three content placement techniques, namely, GCA, CPF, and random caching (RC). In RC, contents stored in clusters are randomly chosen from the file library. The most popular files are cached in the CPF scheme, and the GCA works as illustrated in Algorithm 1.
We can see that the average delay for all content caching strategies increases with $\lambda_k$ since a larger request rate increases the probability of longer waiting time for each request. It is also observed that the CGA, which is locally optimal, achieves significant performance gains over the CPF and RC solutions for the above setup with $N=20$. Fig.~\ref{CPF and GCA} shows that the GCA is superior to the CPF only for small popularity exponent $\beta$. If the popularity exponent $\beta$ is high enough, CPF and GCA will achieve the same performance, since both schemes will cache the most popular files. Additionally, RC fails to reduce the delay as $\beta$ increases, since caching files at random results in a low probability of serving the requested files from local clusters.

\section{Conclusion}
In this work, we propose a novel D2D caching architecture to reduce the network average delay. We study a cellular network where users in the same cluster can exchange cached content with other nearby users via D2D communication; additionally, users in different clusters can cooperate by exchanging their cached content via cellular transmission. We formulate the network average delay minimization problem in terms of cache placement. We study two types of caching policies, namely, caching popular files and greedy caching schemes. By formulating the delay minimization problem as a minimization of a non-increasing supermodular function subject to partition matroid constraints, we show that it can be solved using the proposed GCA scheme within a factor of 2 of the optimum. Numerical results show that the network average delay can be reduced by around 45\% to 80\% by allowing inter-cluster cooperation.
\begin{appendices}
\section{Proof of lemma 1}
We define the ground set that describes the cache placement elements in all clusters as,
\begin{equation}
\mathcal{S} = \{s_1^1, ...,s_k^f, ..., s_k^F, ..., s_K^1, ...,  s_K^F\}
\label{set eqn}
\end{equation}
where $s_k^f$ is an element denoting the placement of file $f$ into the $VCC$ of cluster $k$. This ground set can be partitioned into $K$ disjoint subsets $\{S_1, S_2, ..., S_K\}$, where $S_k =  \{s_k^1, s_k^2, ..., s_k^F\}$ is the set of all files that might be placed in the $VCC$ of cluster $k$.

Let us express the cache placement by the adjacency matrix $\textbf{X} = [x_k,_f ]_{K\times F} \in \{0, 1\}_{K\times F} $. Moreover, we define the corresponding cache placement set $A \subseteq \mathcal{S}$ such that $s_k^f \in A$ if and only if $x_{k,f}=1$. Hence, the constraints on the cache capacity of the $VCC$ of cluster $k \in {\mathcal{K}}$ can be expressed as $A \subseteq \mathcal{S}$, where
\begin{equation}
\mathcal{H} = \{A\subseteq \mathcal{S}: |A\cap S_k| \leq N  \textrm{ for  all}\ k = 1.....K\}
\label{matroid eqn}
\end{equation}
The above expression is derived directly from the constraint that the maximum cache size per cluster is $N$ files, i.e., $\sum_{f=1}^{F}x_k,_f = N$. Comparing $\mathcal{H}$ in (\ref{matroid eqn}) with the definition of partition matroid in (\ref{mat defn eqn}), it is clear that our constraints form a partition matroid with $l = K$ and $k_i = N$. This proves Lemma 1.

\section{Proof of lemma 2}
We consider two cache placement sets $A$ and  $A'$, where $A \subset A'$. For a certain cluster $k \in \mathcal{K}$, we consider adding the caching element $s_i^f$ $\in \mathcal{S}\setminus A'$ to both placement sets. This means that a file $f$ is added to cluster $i$, where the corresponding cache placement element has not been placed neither in $A$ nor in $A'$. The marginal value of adding an element $s_i^f$ to a set is defined as the change in the file download time after adding this element to the set. The average download time for a file $f$ with mean size $\overline{S}$ is $\frac {\overline{S}}{R_{D}}$, $\frac {\overline{S}}{R_{WL}/N_a}$, or $\frac {\overline{S}}{R_{BH}/N_b}$ if the file is obtained from the local cluster, a randomly chosen remote cluster, or the backhaul, respectively.
For our work, we assume that $\frac{R_{WL}}{{N_a}} > \frac{R_{BH}}{{N_b}}$ always hold. For the sake of simplicity, we replace $\frac{R_{WL}}{{N_a}}$ and $\frac{R_{BH}}{{N_b}}$ with their averages, $\overline{R_{WL}}$ and $\overline{R_{BH}}$, respectively. Now, the aggregate transmission rate assumption is $R_D > \overline{R_{WL}} > \overline{R_{BH}}$.

For $D_k$ in (\ref{T eqn}) to be a supermodular function, the difference in the marginal values between the two sets $A$ and $A'$ must be non-positive.
For a user $u$ belonging to cluster $k$ and requesting content $f \in \mathcal{F}$, we distinguish between these different cases:

\begin{enumerate}

  \item According to placement $A'$, user $u$ obtains file $f$ from a remote cluster $j'$, i.e., $s_{j'}^f \in A'$ and $j' \neq k$. In this case, the marginal value with respect to $A'$ is
\begin{align}
G(A' \cup \{s_i^f\}) - G(A') = 0 						
\end{align}
  According to placement $A$, user $u$ obtains file $f$ from a remote cluster $j$, i.e., $s_{j}^f \in A$, again the marginal value is zero. However, if $s_{j}^f\notin A$, the marginal value is given by,
 \begin{align}
G(A \cup \{s_i^f\}) - G(A) = P_{k,f} \Big( \frac {\overline{S}}{\overline{R_{WL}}} -  \frac {\overline{S}}{\overline{R_{BH}}}\Big) 
\end{align}
\item In this case, we assume that $s_i^f = s_k^f$, i.e., the requested file $f$ is cached in cluster $k$. According to placement $A'$, user $u$ obtains file $f$ from the local cluster $k$. Hence, the marginal value is given by,
 \begin{align}
G(A' \cup \{s_i^f\}) - G(A') = P_{k,f} \Big( \frac {\overline{S}}{R_{D}} -  \frac {\overline{S}}{\overline{R_{WL}}}\Big)	
\end{align}
According to placement $A$, user $u$ obtains file $f$ from a remote cluster $j$ when $s_{j}^f \in A$, again the marginal value is given by,
  \begin{align}
G(A \cup \{s_i^f\}) - G(A) = P_{k,f} \Big( \frac {\overline{S}}{{R_{D}}} -  \frac {\overline{S}}{\overline{R_{WL}}}\Big)	
\end{align}
However, if $s_{j}^f\notin A$, the marginal value is written as,
 \begin{align}
G(A \cup \{s_i^f\}) - G(A) = P_{k,f} \Big( \frac {\overline{S}}{\overline{R_{D}}} -  \frac {\overline{S}}{\overline{R_{BH}}}\Big) 
\end{align}
\end {enumerate}
Accordingly, the difference of marginal values between $A$ and $A'$ in all cases is
\begin{align}
G(A \cup \{s_i^f\}) - G(A) - (G(A' \cup \{s_i^f\}) - G(A')) \leq 0				
\end{align}
It is clear that $f(A) \leq f(A')$ for $A \subseteq A'  \subseteq \mathcal{S}$, or equivalently, $f(A) - f(A') \leq 0$. From the definition of supermodularity, it is clear that the delay per request in cluster $k$ is a supermodular set function. Since the sum of supermodular functions is also supermodular, it is enough to prove that the network average delay $D$ in (\ref{optimize eqn}) is a supermodular function. For the monotone non-increasing property, it is intuitive to see that the delay will never increase by caching new files. Hence, Lemma 2 proved that problem (\ref{optimize eqn}) is a monotonically non-increasing supermodular set function minimized under partition matroid constraints.
\end{appendices}

\bibliographystyle{IEEEtran}
\bibliography{bibliography}
\end{document}